\pdfoutput=1
\newif\ifdraft \draftfalse 

\drafttrue

\documentclass{article}
\usepackage{graphicx} 
\usepackage[numbers]{natbib}
\usepackage{fullpage}

\usepackage{macros}
\usepackage{xspace}
\usepackage{xcolor}
\definecolor{darkblue}{rgb}{0,0,.5}
\usepackage[colorlinks=true,allcolors=darkblue]{hyperref}
\usepackage{bbm}
\usepackage{tikz}
\usetikzlibrary{matrix}
\usepackage{overpic}
\usepackage{makecell}
\usetikzlibrary{patterns}
\usepackage{authblk}
\usepackage{enumerate}
\usepackage{amssymb}
\usepackage{amsbsy}
\usepackage{amsmath}
\usepackage{amsthm}
\usepackage{amsfonts}
\usepackage{algorithm}
\usepackage{algpseudocode}
\usepackage{multirow}

\usepackage{latexsym}
\usepackage{graphicx}
\usepackage{color}
\usepackage{xifthen}

\usepackage{color-edits}
\addauthor{Ryan}{red}

\usepackage{latexsym}

\usepackage{subcaption}

\newcommand{\R}{\mathbb{R}}

\newcommand{\N}{\mathbb{N}}

\newtheorem{theorem}{Theorem}

\newtheorem{definition}{Definition}[section]

\title{Private Count Release: A Simple and Scalable Approach for Private Data Analytics}
\author{Ryan Rogers}

\begin{document}

\maketitle

\begin{abstract}
We present a data analytics system that ensures accurate counts can be released with differential privacy and minimal onboarding effort while showing instances that outperform other approaches that require more onboarding effort.  The primary difference between our proposal and existing approaches is that it does not rely on user contribution bounds over distinct elements, i.e. $\ell_0$-sensitivity bounds, which can significantly bias counts.  Contribution bounds for $\ell_0$-sensitivity have been considered as necessary to ensure differential privacy, but we show that this is actually not necessary and can lead to releasing more results that are more accurate.  We require minimal hyperparameter tuning and demonstrate results on several publicly available dataset. We hope that this approach will help differential privacy scale to many different data analytics applications.  
\end{abstract}

\section{Introduction}

There have been several practical deployments of differential privacy for data analytics tasks, including open source libraries and systems from Tumult Labs \cite{Tumult22},\footnote{\url{https://www.tmlt.io/}} OpenDP,\footnote{\url{https://opendp.org/}} and Google.\footnote{\url{https://github.com/google/differential-privacy}}  Notable practical deployments of differential privacy have been announced from Apple \cite{ApplePrivacy17}, the U.S. Census \cite{Census22}, Google \cite{GoogleMobility}, LinkedIn \cite{RogersLinkedIn21}, Meta \cite{FacebookDP}, Microsoft \cite{DingKuYe17}, and Wikipedia \cite{Wikipedia23}.  Although there is a common suite of differentially private algorithms across these deployments, there is a significant onboarding process to include these algorithms, typically tuning many parameters, e.g. sensitivity, that would not be present without differential privacy, and determining where these algorithms should be included in an existing data flow.  \citet{Wikipedia23} describes a standardized workflow for differentially private data releases, which includes a \emph{tuning} stage where parameters are tuned and accuracy is evaluated.  It is this stage that is hard to automate and typically requires experts in the area to determine the correct parameters.  Existing differentially private systems allow their systems to \emph{scale} to very large datasets, which can contain many columns and many rows, which need not be stored in memory.  However, this framework cannot scale to many different of use cases, as each requires some expert as part of the onboarding process.  In this work, we describe a simple approach that allows us to onboard to multiple use cases for various different scenarios.  This system does not allow for arbitrary types of queries, but does handle common tasks based on counts.  In particular, we consider releasing as many counts as possible from a dataset subject to a relative accuracy guarantee for each count and an overall differential privacy guarantee.  This common task has proven to be very difficult for differential privacy, but with recent advances, we are able to handle this task fairly easily.  

As a direct comparison, we consider the ``naive" approach described in the Plume differential privacy system from Google \cite{AminGiJoKuVa22}, which achieves the best accuracy in their work while sacrificing efficiency compared to their proposed solution.  They describe an approach that involves contribution bounding, calling it a necessary step.  We want to distinguish that there are two types of contributions that must be bound: (1) the $\ell_0$-sensitivity or cross-partition contribution bound, meaning how many distinct counts a user can modify in the dataset; (2) the $\ell_\infty$-sensitivity, or per-partition contribution bound, meaning how much a user can modify a single count.  Their approach consists of the following steps:
\begin{itemize}
\item Bound the number of contributions $m$ that any user in the dataset $x$ can contribute, call the resulting dataset $x_m$.
\item Perform a differentially private algorithm to discover the domain $D$ of elements in the dataset $x_m$.  There are several algorithms to do this including \cite{KorolovaKeMiNt09, GopiGuKuShShYe20, SwanbergDeHa23}, all of which require $m$ and $x_{m}$ as well as privacy parameters.
\item Restrict the original dataset $x$ to include only items that were discovered in $D$.  Call this dataset $x^D$.
\item Bound the number of contributions $m'$ that any user in $x^D$ can contribution, call the resulting dataset $x^D_{m'}$. Note that $m$ and $m'$ need not be equal,  
\item Apply a noise addition mechanism, e.g. Gaussian mechanism, to each item in $x^D_{m'}$.  
\end{itemize}

This contribution bounding has two negative side effects: (1) it introduces new parameters $m, m'$ that requires tuning to balance bias and variance of noise beyond the overall privacy parameters, (2) it can significantly slow down data flows.  The Plume system, presents a system in the MapReduce framework that can help with the impact to scalability, but not the parameter tuning.  Perhaps hidden in each differential privacy deployment is the privacy that is consumed in making decisions about these hyperparameters.\footnote{Tumult Labs \url{https://www.tmlt.io/solutions} makes this explicit in their documentation in describing the ``Deploy" phase, which calls out that there are ``unsafe summaries" before the ``privacy engine" on a sensitive dataset.}  Ideally, we would have a differential privacy system that requires no tuning at all, modulo that the overall privacy budget is decided upon before even looking at the data.  Note that the $\ell_\infty$-sensitivity bound can also be enforced when releasing the counts with a noise addition mechanism in the last step described above.  It is useful to know the number of distinct users that have particular items, so we will assume the $\ell_\infty$-sensitivity is 1 in this work, but can easily be extended to include the $\ell_\infty$-sensitivity contribution bound.  This work will describe one such approach that can be a way to implement differential privacy for certain, albeit common, data analytics tasks.  In particular, our goal is to return as many results from a dataset with some target relative error, subject to differential privacy.  Similar success metrics were presented in the case study of releasing page view demographics from Wikimedia using Tumult Labs, quoting that ``more than 95\% of counts had a relative error below 50\%."\footnote{\url{https://www.tmlt.io/resources/publishing-wikipedia-usage-data-with-strong-privacy-guarantees}}

\section{Preliminaries}

We now define approximate differential privacy which depends on neighboring datasets $x, x' \in \cX$, denoted as $x\sim x'$ that differ in the presence or absence of one user's records.

\begin{definition}[\citet{DworkMcNiSm06, DworkKeMcMiNa06}]
An algorithm  $A : \cX \rightarrow \cY$ is $(\epsilon, \delta)$-differentially private if, for any measurable set $S \subseteq \cY$ and any neighboring inputs $x \sim x'$, 
\begin{equation}
\label{eq:dp}
\Pr[A(x) \in S] \leq e^\epsilon \Pr[A(x') \in S] + \delta.
\end{equation}
If $\delta = 0$, we say $A$ is $\diffp$-DP or simply \emph{pure DP}.  
\end{definition}

It will be important to define the sensitivity of a statistic, which is precisely what contribution bounding is attempting to bound.  The $\ell_p$-sensitivity of a statistic $f:\cX \to \R^d$ that takes a dataset $x \in \cX$ to a real vector in $\R^d$ as the following where the max is taken over neighboring $x, x'\in \cX$
\[
\Delta_p(f) = \max_{x \sim x'} \left\{ ||f(x) - f(x')||_p \right\}.
\]

We now define approximate concentrated differential privacy (CDP),\footnote{Although \cite{BunSt16} defines zCDP, to differentiate between CDP from \cite{DworkRo16}, we will use CDP to be the version from \cite{BunSt16}}.  Similar to approximate DP, it permits small probability events of unbounded R\'enyi divergence. 

\begin{definition}[\citet{BunSt16, PapernotSt22}]\label{def:approxazcdp}
Suppose $A: \cX \to \cY$ and $\rho, \delta\geq 0$. We say the algorithm $A$ is $\delta$-approximate  $\rho$-CDP if, for any neighboring datasets $x, x'$, there exist distributions $P', P'', Q', Q''$ such that the outputs are distributed according to the following mixture distributions:
\begin{align*}
 A(x) \sim (1 - \delta)P' + \delta P''  \qquad A(x') \sim (1 - \delta)Q' + \delta Q'',
\end{align*}
where for all $\lambda\geq 1$, $D_\lambda(P' \| Q') \leq \rho \lambda$ and $D_\lambda(Q' \| P') \leq \rho \lambda$.
\end{definition}

We can also convert approximate differential privacy to approximate CDP and vice versa.
\begin{theorem}[\citet{BunSt16}]
\label{thm:CDPtoDP}
If $A$ is $(\diffp, \delta)$-DP then it is $\delta$-approximate $\diffp^2/2$-CDP.  If $A$ is $\delta$-approximate $\rho$-CDP then it is $(\rho + 2\sqrt{\rho \log(1/\delta')}, \delta' + \delta)$-DP for any $\delta'>0$.  
\end{theorem}

The classical CDP mechanism is the Gaussian Mechanism.  Note that the Gaussian Mechanism was originally introduced as satisfying approximate DP, but it was then shown to satisfy pure CDP in later work \cite{DworkRo16, BunSt16}.  

\begin{theorem}[\citet{BunSt16}]
\label{thm:Gaussian}
Let $f: \cX \to \R^d$ have $\ell_2$-sensitivity $\Delta_2(f)$, then the mechanism $M: \cX \to \R^d$ where $M(x) = f(x) + (Z_1, \cdots, Z_d)$ with $\{Z_i \} \stackrel{i.i.d.}{\sim} \mathrm{N}(0, \tfrac{\Delta_2(f)^2}{2\rho})$ is $\rho$-CDP for $\rho>0$.
\end{theorem}

We also state the composition property of CDP, showing that the overall privacy parameters degrade after multiple CDP algorithms are run on a dataset.
\begin{theorem}
\label{thm:compCDP}
Let $A_1: \cX \to \cY$ be $\delta_1$-approximate $\rho_1$-CDP and $A_2: \cX \times \cY \to \cZ$ where $A_2(\cdot, y)$ is $\delta_2$-approximate $\rho_2'$-CDP for all $y \in \cY$.  Then $A: \cX \to \cZ$ where $A(x) = A_2(x, A_1(x))$ is $(\delta_1 + \delta_2 - \delta_1 \cdot \delta_2)$-approximate $\rho_1 + \rho_2$-CDP.
\end{theorem}

\section{Private Count Release}

Rather than use a \emph{partition selection} DP algorithm, we will instead utilize the Unknown Domain Gumbel mechanism from \cite{DurfeeRo19} to iteratively find the highest counts.  The Unknown Domain Gumbel algorithm in Algorithm~\ref{alg:UnkGumbel} takes a few parameters, beyond the typical privacy parameters: the maximum number of things we want to return $k$ and the total ranked counts that we have access to $\bar{k} > k$. One can imagine that as a first step we would compute a histogram over all possible items in the dataset with their corresponding counts.  We will assume that a user can contribute to a single count by at most 1 ($\ell_\infty$-sensitivity is 1) but a user can modify as many as the maximum possible number of counts, $\bar{k}$.    This is easily achieved by running a distinct count aggregate function by a user ID column, rather than a doing COUNT(*) in SQL and ordering by the distinct counts in descending order up to $\bar{k}$.
\begin{algorithm}
\caption{Unknown Domain Gumbel from Top-$(\bar{k}+1)$}\label{alg:UnkGumbel}
\begin{algorithmic}
    \Require Histogram $h$, noise scale $\beta>0$, threshold $T>0$, and $k, \bar{k}$
    \Ensure Sorted list of at most $k$ items, $I_k$.
    \State Let $h_{(\bar{k})}$ be the histogram consisting of the top-$(\bar{k})$ items 
    \State Let $c_{(\bar{k}+1)}$ be the count of the $(\bar{k} + 1)$-th item in $h$.
    \State Set $\tilde{T} = T + c_{(\bar{k}+1)}+\mathrm{Gumb}(\beta)$
    \State Initialize $\tilde{h} = \emptyset$.
    \For{each item $i$ where $(i, c_i) \in h_{(\bar{k})}$ such that $c_i >0$}
        \State Set $\tilde{c}_i = c_i+ \mathrm{Gumb}(\beta)$
        \If{$\tilde{c}_i > \tilde{T}$ }
            \State $\tilde{h} = \tilde{h} \cup \left\{ (i, \tilde{c}_i) \right\}$
        \EndIf
    \EndFor
    \State Let $I_k$ be the ordered list of at most $k$ items that are sorted in descending order by their count in $\tilde{h}$.  
    \If{$I_k<k$}
        \State $I_k = I_k \cup \{ \bot \}$
    \EndIf
\end{algorithmic}
\end{algorithm}
We then state Unknown Domain Gumbel's privacy guarantee.
\begin{theorem}[\citet{DurfeeRo19, Rogers23}]
\label{thm:unkGumbel}
    Assume that a user can modify all counts in an input histogram $h$, but can only change a single count by at most 1.  If we use the noise scale $\beta = 1/\diffp$, and threshold 
    \[
    T =  1 + \frac{1 }{\diffp} \log(\tfrac{\bar{k}	}{\delta}),
    \] 
    then Algorithm~\ref{alg:UnkGumbel} is $\delta$-approximate $k\diffp^2/8$-CDP.
\end{theorem}

We will use the Unknown Domain Gumbel algorithm to iteratively find the element with the highest count, and then add noise to that element's count.  Each of these steps will need to be accounted for in the overall privacy budget, ensuring that we do not return too many counts and hence exceeding our budget.  Note that we can return the $\bot$ element, which is not very informative and tells us that our privacy parameters are too low, i.e. the threshold is too high.  If there is privacy budget remaining, we would like to reduce the noise so that we no longer get a $\bot$ element.  This natural approach requires adaptively picking the privacy parameters, which can be handled with a \emph{privacy filter} \cite{RogersRoUlVa16}.  Traditional differential privacy requires setting all the privacy parameters in advance, but a privacy filter allows us to adaptively select privacy parameters based on prior outcomes is such a way so that we do not exceed our overall privacy budget.  Recent work from \cite{WhitehouseRaRoWu23} shows that we can essentially use traditional DP composition bounds despite the parameters being adaptively selected.  

\begin{theorem}[\citet{WhitehouseRaRoWu23}]
\label{thm:azcdp}
Let $(A_n: \cX \to \cY)_{n \geq 1}$ be an adaptive sequence of algorithms. Assume that $\delta_n,\rho_n : \cY^{n-1} \to \R_{\geq 0}$ and furthermore $\sum_{n=1}^\infty \delta_n(y_{1:n-1}) \leq \delta$ for all $y = (y_1, \cdots, y_k)$.
For any $n \geq 1$, assume that $A_n(\cdot; y_{1:n-1})$ is $\delta_n(y_{1:n-1})$-approximate $\rho_n(y_{1:n-1})$-CDP for any prior outcomes $y_{1:n-1}$. We define the function $N: \cY^\infty \to \N$ where
\[
N(y_1, y_2, \cdots) = \inf \left\{n : \sum_{\ell = 1}^{n+1} \rho_\ell(y_{1:\ell-1}) > \rho\right\}.
\]
Then $A_{1:N(\cdot)}(\cdot)$ is $\delta$-approximate $\rho$-CDP, where $N(x) = N( (A_{n}(x))_{n \geq 1} )$ is a stopping time relative to the natural filtration generated by $A_1(x), A_2(x), \cdots$.  Furthermore, $A_{1:N(\cdot)}(\cdot)$ is $(\rho + 2\sqrt{\rho \log(1/\delta')}, \delta + \delta')$-DP for any $\delta' >0$.  
\end{theorem}

This type of composition will allow us to start with a very small privacy parameter $\diffp_1$ and then if the Unknown Domain Gumbel algorithm returns $\bot$ for the top-$1$ query, we then increase $\diffp_1$ to $\diffp_2$.  We then continue in this way until we actually get an element that is not $\bot$.  However, we will have to accumulate the privacy losses that were consumed for each unsuccessful iteration of Unknown Domain Gumbel.  We then adopt the doubling approach, where $\diffp_{i+1} = \sqrt{2} \diffp_i$, see \cite{LigettNeRoWaWu17, RogersSaWuRa23} for more details.  

Once we have selected a non-$\bot$ element, we then want to use the Gaussian Mechanism on the count for that discovered element.  Note that the count has $\ell_\infty$-sensitivity of 1, which is also the $\ell_2$-sensitivity since it is a one dimensional count.  The question then becomes, what noise level should we use?  We could just use the last $\diffp_i$ parameter we used when we found that element, but can  this be done smarter?  More generally, if you have a target relative error $r$, one can use the fact that the released item from Unknown Gumbel most likely has a true count $c$ of at least $ 1 + \log(\bar{k}/\delta)/\diffp_i$, to then use $\sigma(\diffp_i)$ standard deviation noise to the count of the selected item, where we want $\tfrac{c+ 1.5 \sigma(\diffp_i)}{c} \leq (1 + r)$, so that
\begin{equation}
\sigma(\diffp_i) = \left(\frac{r}{1.5} \right) \cdot \left(1 + \frac{ \log(\bar{k}/\delta_*))}{\diffp_i} \right).
\label{eq:SigmaNoise}
\end{equation}
Note that we can use any constant in front of $ \sigma(\diffp_i)$ depending on the accuracy requirements.  After we have an element and its noisy count, we then repeat the whole process again until we have exhausted the given privacy budget $\rho,\delta$.  Note that it is possible to improve this approach further by using the \emph{Brownian Noise Reduction} \cite{WhitehouseWuRaRo22}, rather than using the Gaussian mechanism to adaptively pick the largest noise level that would still provide a target relative error.  The privacy gets accumulated in the same way as if a Gaussian mechanism were used with that noise level chosen in advance \cite{RogersSaWuRa23}.  

We present our overall approach, Private Count Release (PCR), in Algorithm~\ref{alg:pcr}.  We want to point out that this whole process can be run with a single dataframe stored in memory that can be computed offline once.  The initial step of creating a histogram, or dataframe, with elements and their corresponding counts might require Spark, but that needs to only be accessed once to create a dataframe of at most $\bar{k}$ rows, not scaling with the number of users.  The Unknown Domain Gumbel and Gaussian mechanisms need not ever access that event level dataset with user IDs, which can be massive.  

\begin{algorithm}
\caption{Private Count Release}\label{alg:pcr}
\begin{algorithmic}
    \Require Histogram $h$, minimum privacy parameters $\diffp_*>0$, $\delta_*>0$, bound $\bar{k}$, and overall privacy parameters $\rho> \diffp_*^2/4$ and $\delta > \delta_*$.
    \Ensure Noisy Histogram $\tilde{h}$ with elements and counts.
    \State Initialize $\rho_\text{curr} = 0$, $\delta_\text{curr} = 0$, $\diffp_{\text{last}} = \diffp_*$, $\tilde{h} = \{ \}$.
    \While{$\rho_\text{curr} + \diffp_{\text{last}}^2/4 \leq \rho$ and $\delta_\text{curr} + \delta_* \leq \delta$}
    	\State $\beta = \frac{1}{\diffp_{\text{last}}}$ 
	\State $T = 1 + \beta \cdot \log(\bar{k}/\delta_*)$
   	\State $s = $ Unknown Domain Gumbel with inputs $(h, \beta, T, k = 1, \bar{k})$
	\State Update $\rho_\text{curr} = \rho_\text{curr} + \diffp_{\text{last}}^2/8$, $\delta_{\text{curr}} = \delta_{\text{curr}}+ \delta_*$.
	\If{$s == \bot$}
		\State $\diffp_{\text{last}} = \sqrt{2} \cdot \diffp_{\text{last}}$.
	\Else
		\State $\sigma = \sigma(\diffp_{\text{last}})$ from \eqref{eq:SigmaNoise} as long as $\sigma > 2/\diffp_{\text{last}}$
		\State Update the result histogram $\tilde{h}[s] = h[s] + \mathrm{N}(0, \sigma^2)$
		\State Remove $s$ from $h$, as it is already discovered.
		\State Update $\rho_\text{curr} = \rho_\text{curr} + \left( 1/\sigma(\diffp_{\text{last}}) \right)^2/2$.
	\EndIf
    \EndWhile
    \end{algorithmic}
\end{algorithm}
We then state overall privacy guarantee.
\begin{theorem}
\label{thm:overallPrivacy}
    Private Count Release in Algorithm~\ref{alg:pcr} is $\delta$-approximate $\rho$-CDP as well as $(\rho + 2 \sqrt{2\log(1/\delta')}, \delta + \delta')$-DP for all $\delta' >0$
\end{theorem}
\begin{proof}
The proof follows from Unknown Domain Gumbel being $\delta_*$-approximate $\diffp_{\text{last}}^2/8$-CDP for each call to it. Then when we select an element that is not $\bot$, we use the Gaussian mechanism, which is $\left( 1/\sigma(\diffp_{\text{last}}) \right)^2/2$-CDP.  We then apply the privacy filters from Theorem~\ref{thm:azcdp} to claim that the while loop ensures we terminate before the privacy budget is exhausted.
\end{proof}

Although there are some nuisance parameters, $\bar{k}$, $\delta_*$, and $\diffp_*$, they have a small impact to the results as long as they are set sufficiently high or low, as we will see in our experiments.  Essentially $\bar{k}$ should be thought of as the larger the better, but that might pose a slowdown in the computations, $\diffp_*$ determines the largest amount of noise a use case is willing to tolerate, and $\delta_*$ can be tied to the maximum number of results $\ell^*$ that might be expected so that $\delta_* \cdot \ell* \leq \delta$.  We note here that we can also return the standard deviation of noise that is added to each released count with no impact to the overall privacy guarantee.

\section{Results}

We now showcase the results using Private Count Release (PCR) in Algorithm~\ref{alg:pcr}.  We will use several datasets to show that our approach works across various settings.  For each dataset, we consider the number of \emph{users}, which might denote a Wikipedia page, news article, or a user on website.  Further, each user will contribute some number of elements, whether it is words, stock symbol, or movie and we provide percentiles on the number of contributions each user has.  Each distinct contribution across all users is then summed to provide the domain size of the dataset.  Table~\ref{table:data} provides some metrics summarizing each dataset.  We also used code that was used to preprocess data from \cite{CarvalhoWaGo22}, which is available online.\footnote{\url{https://github.com/ricardocarvalhods/diff-private-set-union/tree/main}}  

\begin{table}[h]
\centering
\caption{\label{table:data} Number of users in each dataset we use with the percentiles for the number of contributions each user has after deduplicating the items from each user}
\begin{tabular}{| c |c|c|c|c| c| c|}
\cline{4-7}
 \multicolumn{3}{c}{} & \multicolumn{4}{|c|}{\textbf{Contributions Per User}} \\
\hline
\textbf{Data Set} & \textbf{\# Users}& \textbf{Domain Size} & \textbf{50\%-tile} & \textbf{75\%-tile} & \textbf{95\%-tile} & \textbf{99\%-tile}\\
\hline
Finance & 1,400,469 & 6,193 & 1 & 1 & 1& 1 \\
\hline
Reddit Small Sample & 341 & 10,575 & 23 & 50 & 128& 330 \\
\hline
Reddit Full Data & 223,388 & 153,704 & 23 & 56 & 194 & 450 \\
\hline
Wikipedia & 245,103 & 631,855 & 84 & 164 & 390 & 589\\
\hline
Movie Lens & 162,541 & 59,047 & 71 & 162 & 554 & 1228\\
\hline
\end{tabular}
\end{table}

In all our experiments we will vary the privacy $\rho \in \{0.1, 0.5, 1.0\}$ and fix $\delta = 10^{-6}$.  To determine the contribution bound for each user in the Plume approach, we assume that the actual 95th or 99th percentile on the dataset is known.  This parameter should be computed subject to differential privacy, but to reduce the number of parameters we need to tune, such as how much of the privacy budget to allocate for this step, we allow direct access to the percentile.  Hence, our implementation of Plume will directly compute the 95th or 99th percentile of the unique contributions per user and use this as the contribution bound $m$ for the partition selection.  Then with the discovered domain from partition selection, the contribution bound $m'$ is the 95th or 99th percentile of the unique contributions from the discovered domain per user (again, without differential privacy).  We then use $m'$ as the $\ell_2$-sensitivity in the Gaussian mechanism to return the corresponding counts for the items that were discovered.  There is another hyperparameter to set, which is the proportion $\alpha \in (0,1)$ of the privacy budget that we should allow for partition selection, so that the remaining privacy budget can be used in the Gaussian mechanism.  We fix $\alpha = 0.5$ in our experiments.  We also show results when we remove low counts from the outcomes, to reduce high relative error noisy counts, which we call Plume with Threshold.  The threshold we use depends on the target relative error $r$, which we set to 10\% in our experiments.  The threshold $T$ we use is then to ensure that if the noisy counts $\tilde{h} > T$ then $\tilde{h}$ is likely to have small error with $\sigma$ standard deviation noise from the Gaussian mechanism.  
\[
\frac{\tilde{h} + \sigma}{\tilde{h} - \sigma} \leq 1 + r \implies \tilde{h} \geq \frac{(2 + r) \sigma}{r} =: T
\]

In PCR, we fix parameters $\diffp_* = 0.0005$, $\delta_* = 10^{-11}$, and $\bar{k} = 10,000$ in all experiments.  Note that the datasets vary in size between a few hundred to over a million, yet we use the same parameters to show that we need not focus on hyperparameter tuning and we obtain results that we are accurate.  We measure accuracy as the number of results returned (recall) and the relative error on the returned counts (precision).  We define relative error as the following where $c$ is the true count and $\tilde{c}$ is the noisy count.  
\[
\mathrm{Relative Error} = \frac{|c - \tilde{c}|}{c} \cdot 100
\]
We now present results for each dataset.    

\paragraph{Finance:}   This dataset comes from \cite{FinanceData} and is collected from stock market news data and was also used in \cite{SwanbergDeHa23}.  We use the ``raw\_analyst\_ratings.csv" file with the index of the row as the author id, so that the number of users is the same as the number of rows, i.e. 1,400,469.  Rather than use the ``title" column to use the words, we instead use the stock ticker symbol as the column we want to release counts for.  This is a favorable setting for Plume or any approach that uses contribution bounding as each user can contribute a single item or stock ticker symbol in this case.  Using PCR will be overly conservative due to accumulating the privacy loss with each item that is returned although the entire histogram could be returned with a small amount of noise.   In our implementation, Plume uses a contribution bound of 1 when using either the 95th or 99th percentile, which does not impact the true counts of each stock ticker symbol.  Results are presented in Figure~\ref{fig:Finance}.  
\begin{figure}[H]
\centering
\includegraphics[width=0.45\textwidth]{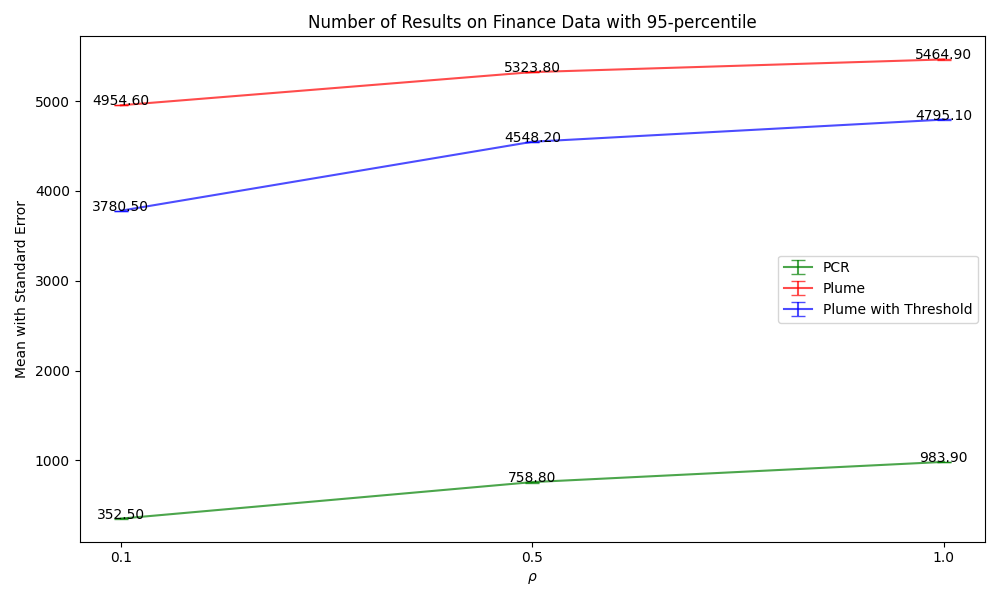}
\includegraphics[width=0.45\textwidth]{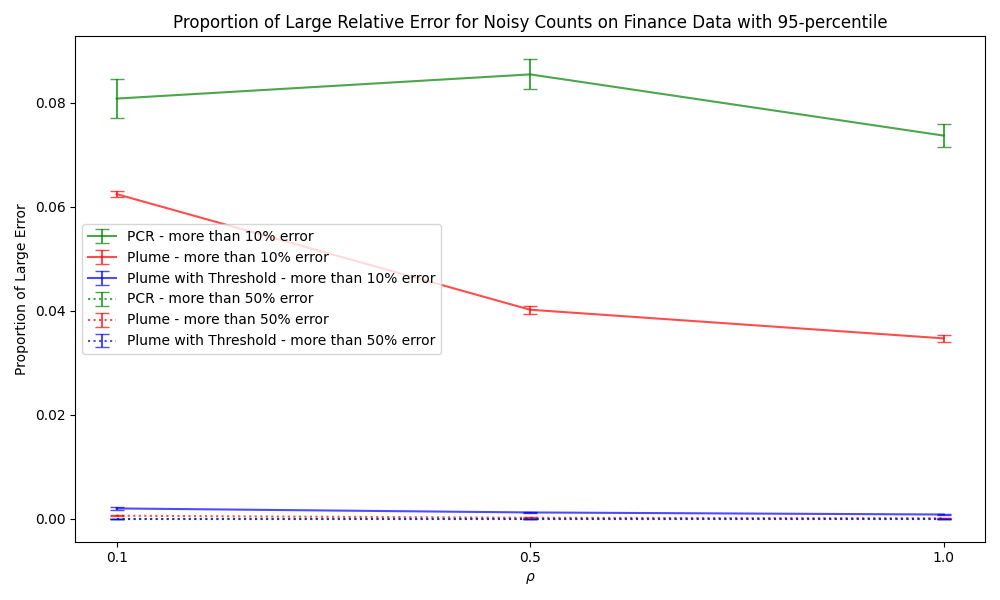}
\caption{\label{fig:Finance} Results for three approaches: PCR, Plume, and Plume with Threshold on the Finance data.  We show recall and precision for $\rho \in \{0.1, 0.5, 1.0 \}$ and $\delta = 10^{-6}$ averaged over 10 independent trials. }
\end{figure}

The key takeaways from these results are that Plume can return a lot more results compared with PCR.  We see that Plume without a threshold returns few of its counts above the 10\% target relative error  PCR returns substantially fewer results and less than 10\% of the noisy counts are beyond the target relative error, which can be adjusted easily.    However, it is quite rare to get such a nice dataset where each user has a single contribution, in which case one can easily use the Laplace or Gaussian noise over positive count items and release only those above a threshold \cite{KorolovaKeMiNt09, WilsonZhLaDeSiGi19}.

\paragraph{Reddit subsample and full data: } The Reddit comments dataset\footnote{\url{https://github.com/heyyjudes/differentially-private-set-union/tree/ea7b39285dace35cc9e9029692802759f3e1c8e8/data}} has been used in previous work, including \cite{GopiGuKuShShYe20, CarvalhoWaGo22, SwanbergDeHa23}.  This data consists of comments from Reddit authors.  We will use both the full dataset and a small sample of the dataset, the first 340 users, to show performance on different scales of data.  To find the most frequent words from distinct authors, we take the set of all distinct words contributed by each author, which can be arbitrarily large and form the resulting histogram which has $\ell_\infty$-sensitivity 1 yet unbounded $\ell_0$-sensitivity.  We then compare our approach with Plume \cite{AminGiJoKuVa22}, which uses contribution bounding as described above, using the actual percentiles on the true data.  Results are presented in Figure~\ref{fig:RedditSample} for the subsample and in Figure~\ref{fig:RedditFull} for the full dataset.  
\begin{figure}[H]
\centering
\includegraphics[width=0.45\textwidth]{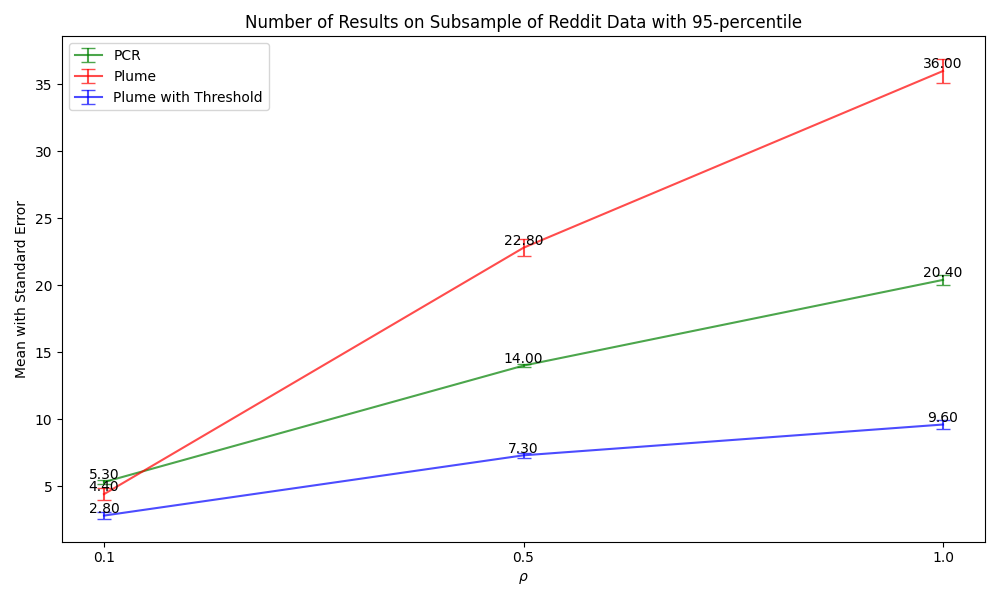}
\includegraphics[width=0.45\textwidth]{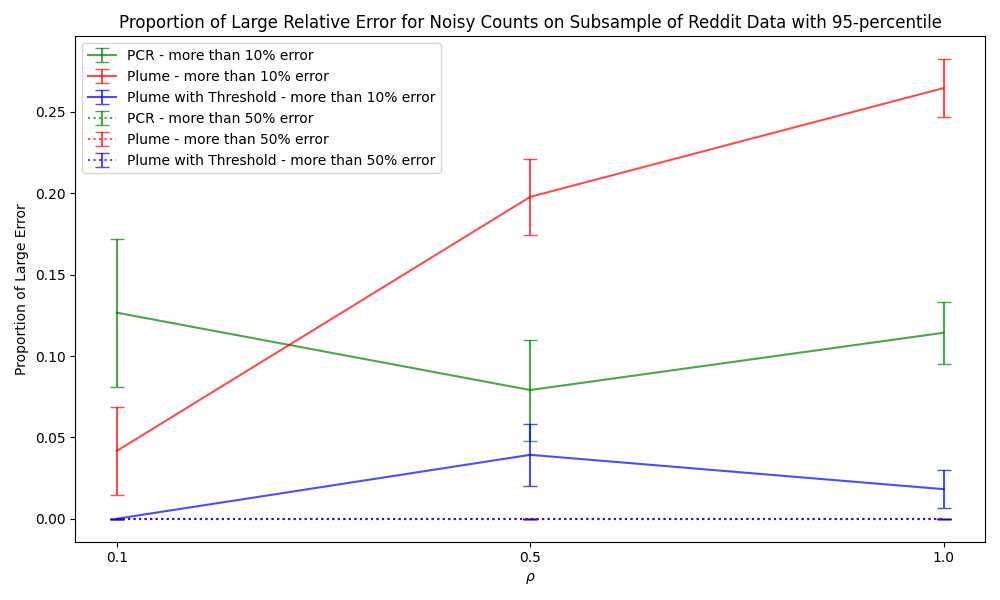}
\includegraphics[width=0.45\textwidth]{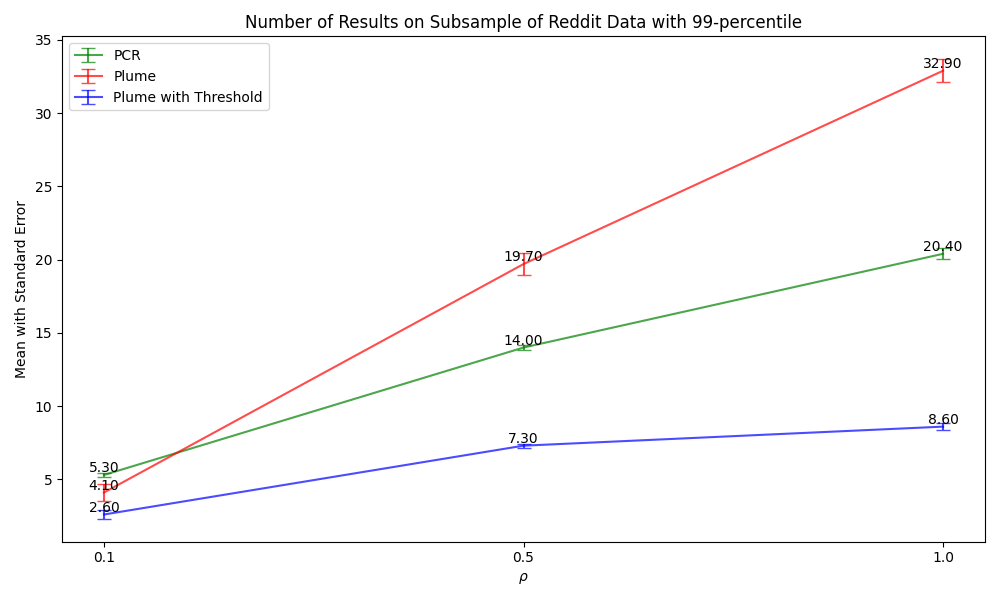}
\includegraphics[width=0.45\textwidth]{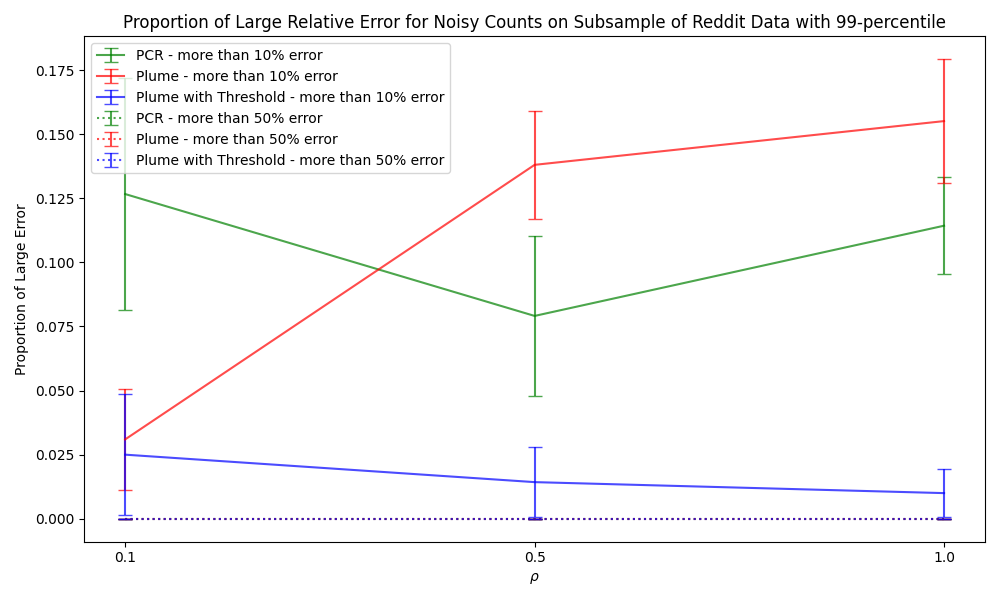}
\caption{\label{fig:RedditSample} Results for three approaches: PCR, Plume, and Plume with Threshold on the Reddit comments subsample data.  We show recall and precision for $\rho \in \{0.1, 0.5, 1.0 \}$ and $\delta = 10^{-6}$ averaged over 10 independent trials. The top plots use the true 95th-percentile for contribution bounding in Plume and the bottom plots use the true 99th-percentile for contribution bounding in Plume.}
\end{figure}

\begin{figure}[H]
\centering
\includegraphics[width=0.45\textwidth]{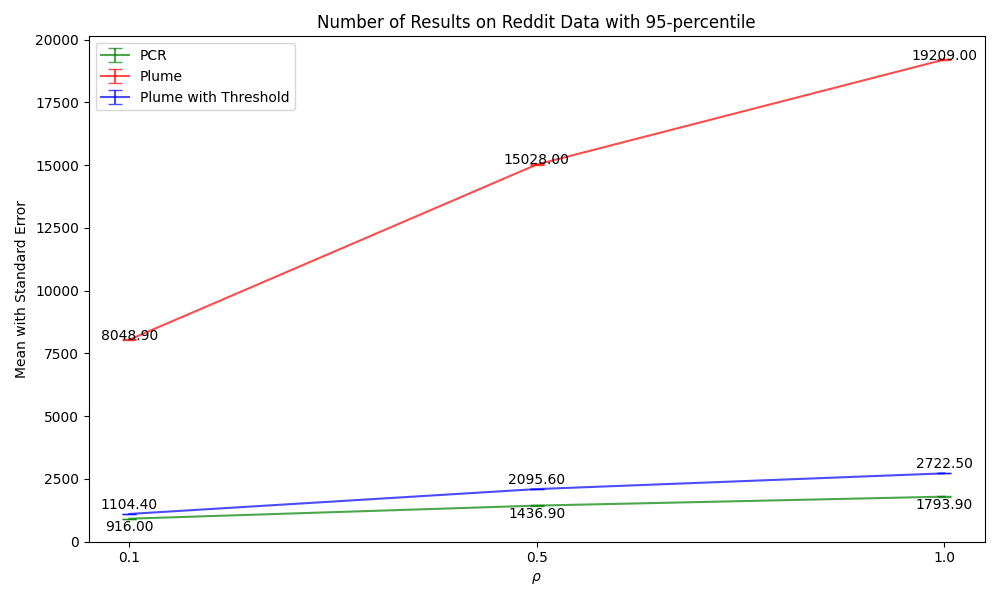}
\includegraphics[width=0.45\textwidth]{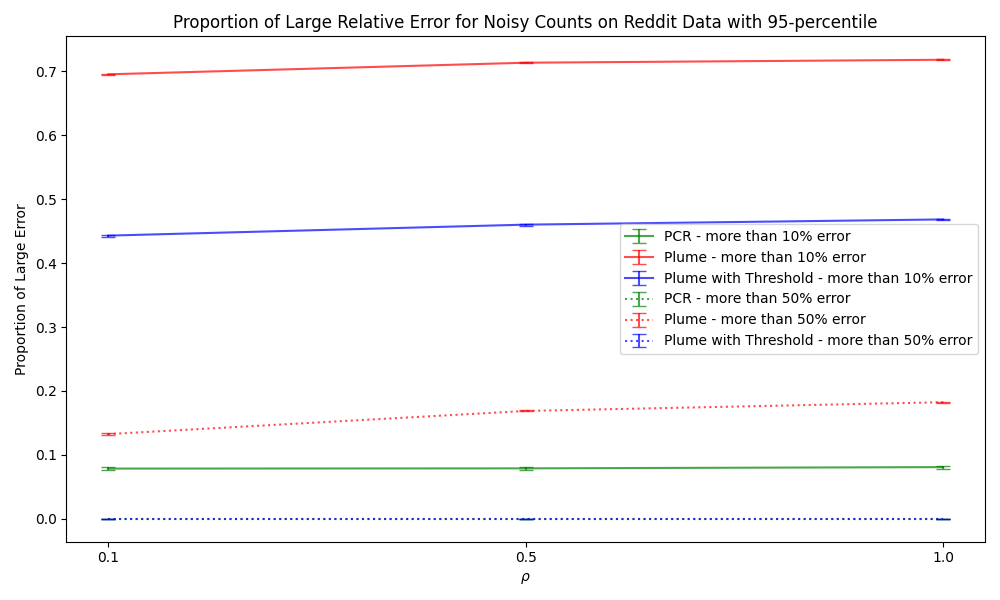}
\includegraphics[width=0.45\textwidth]{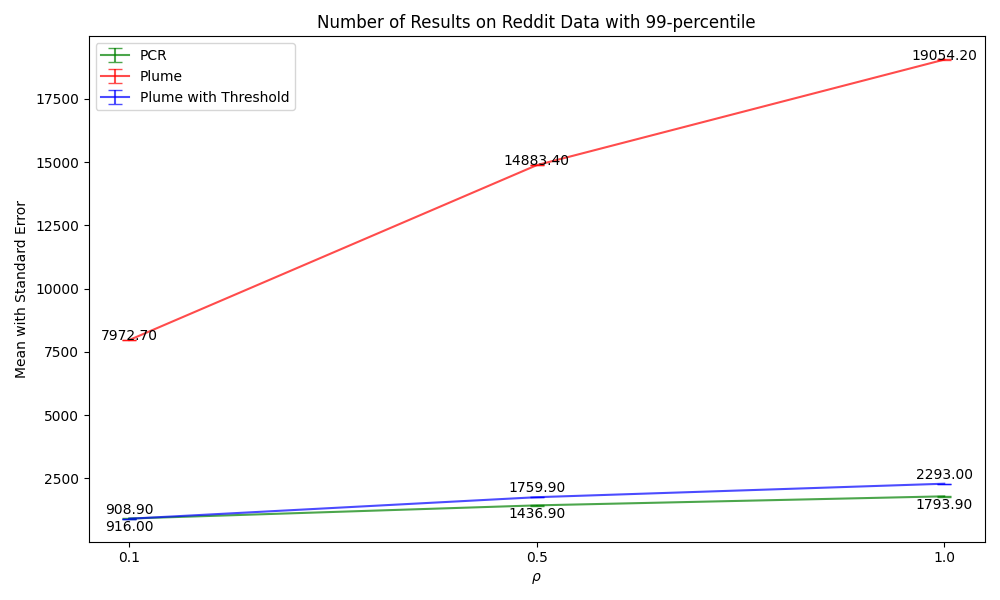}
\includegraphics[width=0.45\textwidth]{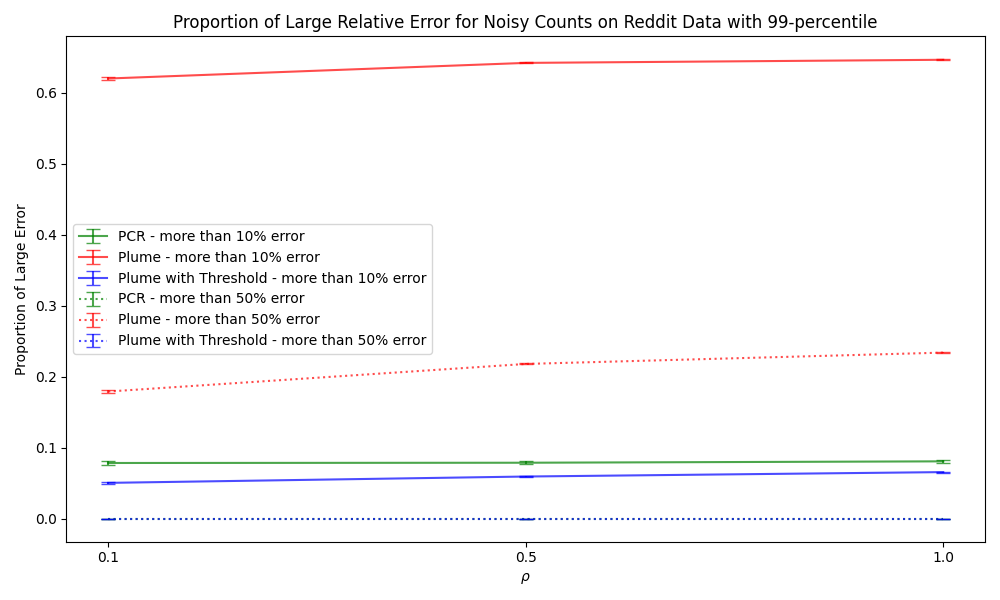}
\caption{\label{fig:RedditFull} Results for three approaches: PCR, Plume, and Plume with Threshold on the Reddit comments full data.  We show recall and precision for $\rho \in \{0.1, 0.5, 1.0 \}$ and $\delta = 10^{-6}$ averaged over 10 independent trials. The top plots use the true 95th-percentile for contribution bounding in Plume and the bottom plots use the true 99th-percentile for contribution bounding in Plume.}
\end{figure}

The key takeaways from these results are that PCR works well with no hyperparameter tuning on vastly different scales of data (from 340 to 223,338 users).  Plume can return a lot more results compared with PCR, but when imposing a threshold, the number of results returned become comparable to PCR.  On the subsample of Reddit data in Figure~\ref{fig:RedditSample}, we see that PCR can return more results than Plume with a threshold while only 10\% of the returned noisy counts in PCR are beyond the 10\% relative error target.  Using either the 95th or 99th percentile in Plume does not seem to impact the results too much, and we see that Plume with a threshold returns almost no results beyond the target relative error, so perhaps a smaller threshold could be used, but this would require tuning.

In Figure~\ref{fig:RedditFull}, we see that the percentile used significantly impacts the relative error from Plume's results.  Even after imposing a threshold, if we use the 95th percentile, almost half of the results from Plume have a relative error beyond the target 10\%. However, when using the 99th percentile, we see that Plume with a threshold returns almost all of its counts with the target relative error.  This shows how important it is to pick an accurate percentile, hence a significant portion of the overall privacy budget should be allocated for this calculation, which we are not accounting for here.  Even if the exact 99th percentile value is used in Plume with a threshold, PCR can return just as many results with the target relative error when the overall privacy loss is small ($\rho = 0.1$).       

\paragraph{Wikipedia: } The Wikipedia dataset comes from \cite{WikiData} and has been used in previous work, including \cite{CarvalhoWaGo22, SwanbergDeHa23}.  This data consists of comments from Wikipedia abstracts, where we treat each Wikipedia page as a user.  To find the most frequent words from distinct users, we take the set of all distinct words contributed by each user, which can be arbitrarily large and form the resulting histogram which has $\ell_\infty$-sensitivity 1 yet unbounded $\ell_0$-sensitivity.  Results  on the Wikipedia data are presented in Figure~\ref{fig:Wiki} .  
\begin{figure}[H]
\centering
\includegraphics[width=0.45\textwidth]{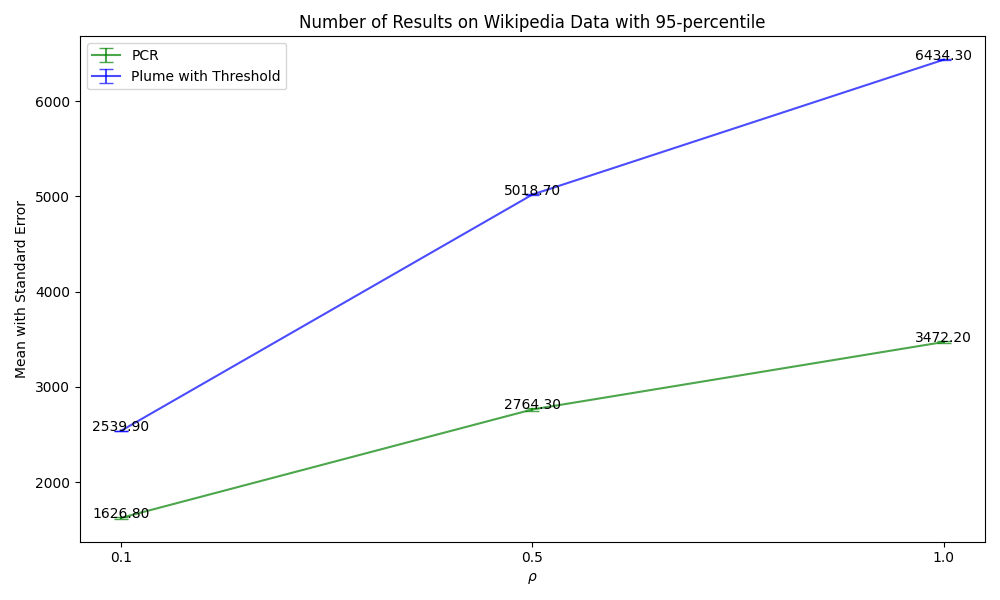}
\includegraphics[width=0.45\textwidth]{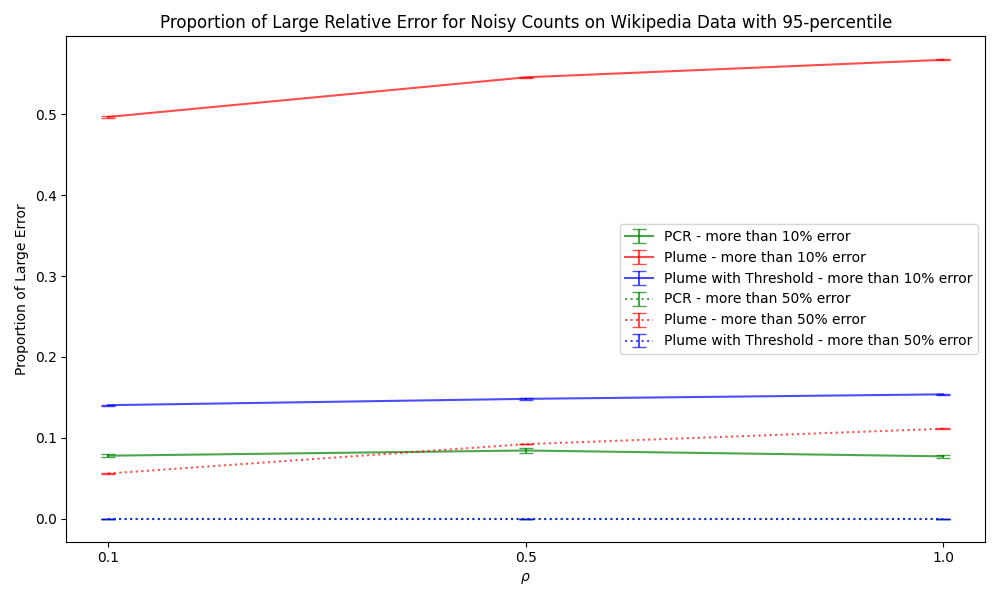}
\includegraphics[width=0.45\textwidth]{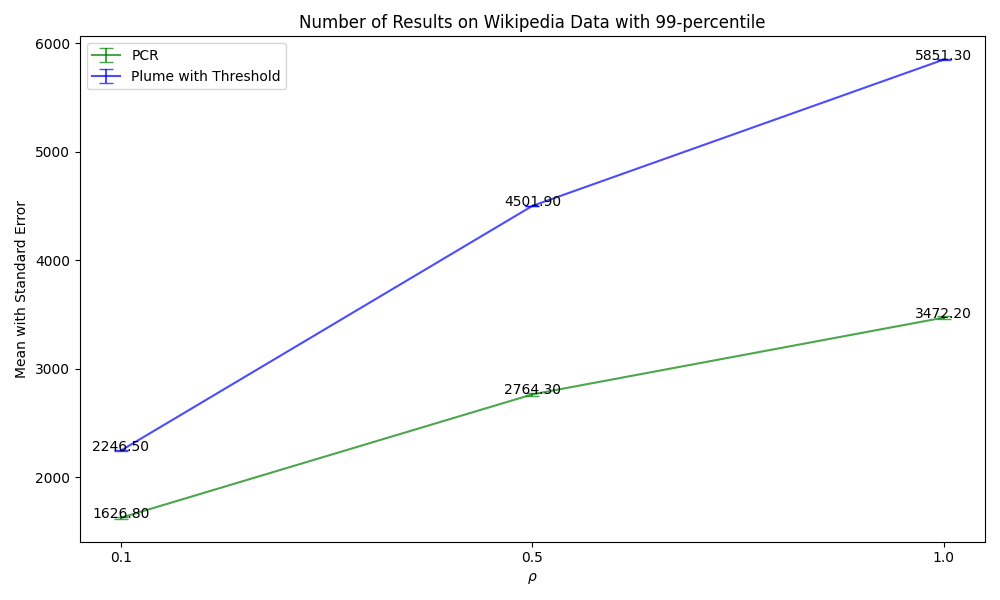}
\includegraphics[width=0.45\textwidth]{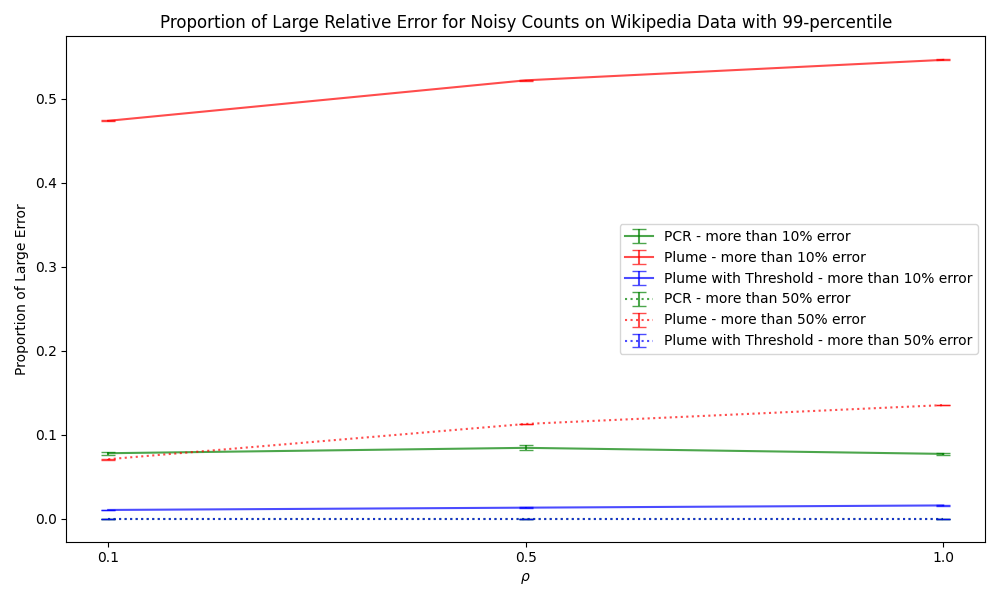}
\caption{\label{fig:Wiki} Results for three approaches: PCR, Plume, and Plume with Threshold on the Wikipedia data.  We show recall and precision for $\rho \in \{0.1, 0.5, 1.0 \}$ and $\delta = 10^{-6}$ averaged over 10 independent trials. The top plots use the true 95th-percentile for contribution bounding in Plume and the bottom plots use the true 99th-percentile for contribution bounding in Plume.}
\end{figure}

We do not show the number of results returned from Plume without a threshold due to the scale being very different from the other approaches, but we write them here.  For the 95th percentile we have: $\rho = 0.1$ we get $12981.5$ results, $\rho = 0.5$ we get $25900.6$ results, and $\rho = 1$ we get $34703.3$ results on average over 10 independent trials.  For the 99th percentile we have: $\rho = 0.1$ we get $13119.9$ results, $\rho = 0.5$ we get $26129.8$ results, and $\rho = 1$ we get $35023.8$ results on average over 10 independent trials.  We have similar takeaways here as in the other datasets, PCR returns fewer results, but almost all are within the target relative error.  We see that Plume with a threshold does return more results than PCR, but note that it still returns more than 10\% of its counts with high relative error when we use the 95th percentile, so getting an accurate percentile is crucial.  

\paragraph{Movie Lens: } The MovieLens dataset comes from \cite{HarperKo15, WikiData} where we use the ``reviews.csv" file to get the movies reviewed by each user to count the number of movies reviewed by all users. Each user has reviewed at least 20 movies.  Results  on the MovieLens data are presented in Figure~\ref{fig:Movie} .  
\begin{figure}[H]
\centering
\includegraphics[width=0.45\textwidth]{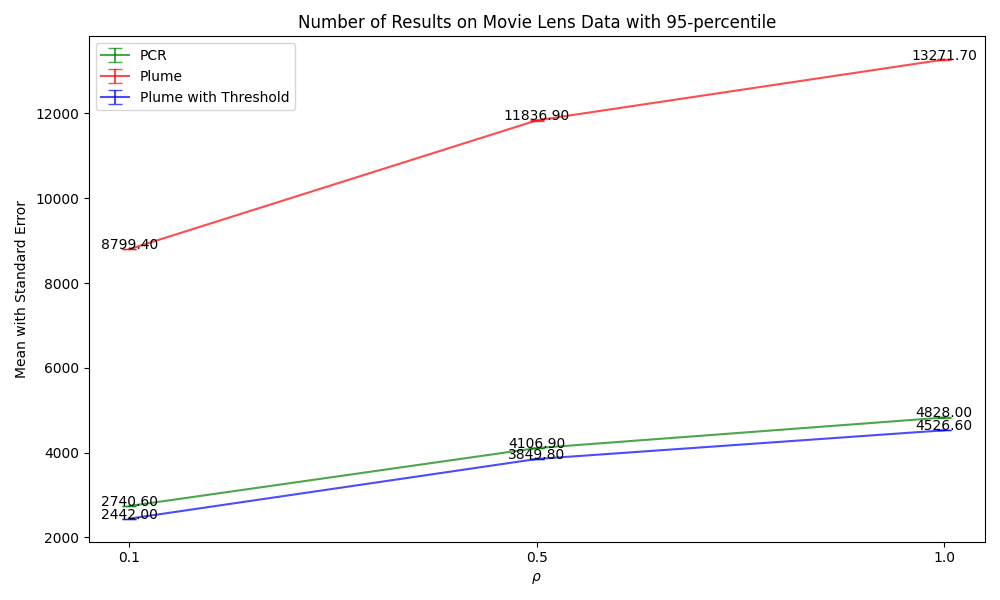}
\includegraphics[width=0.45\textwidth]{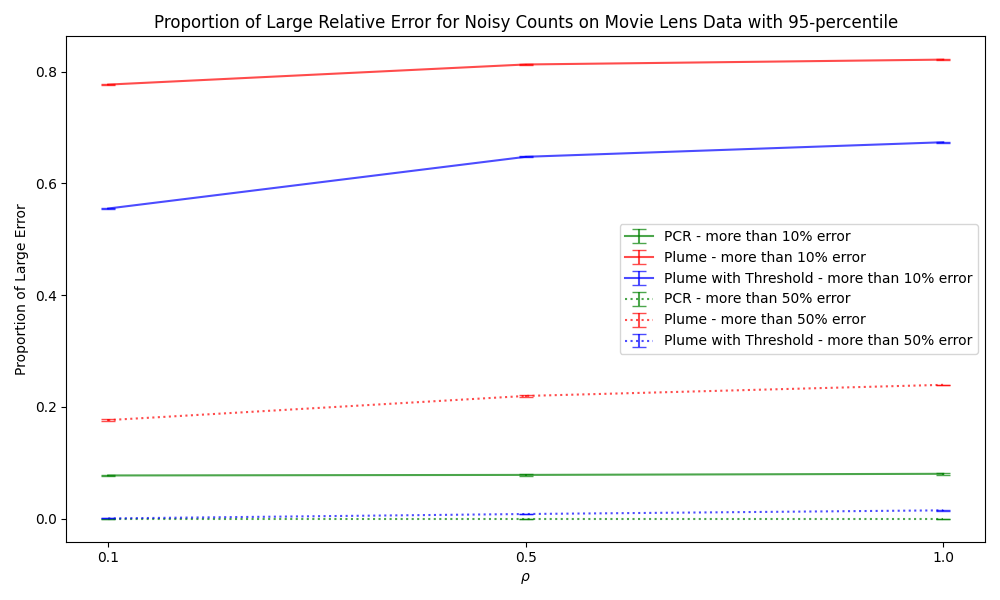}
\includegraphics[width=0.45\textwidth]{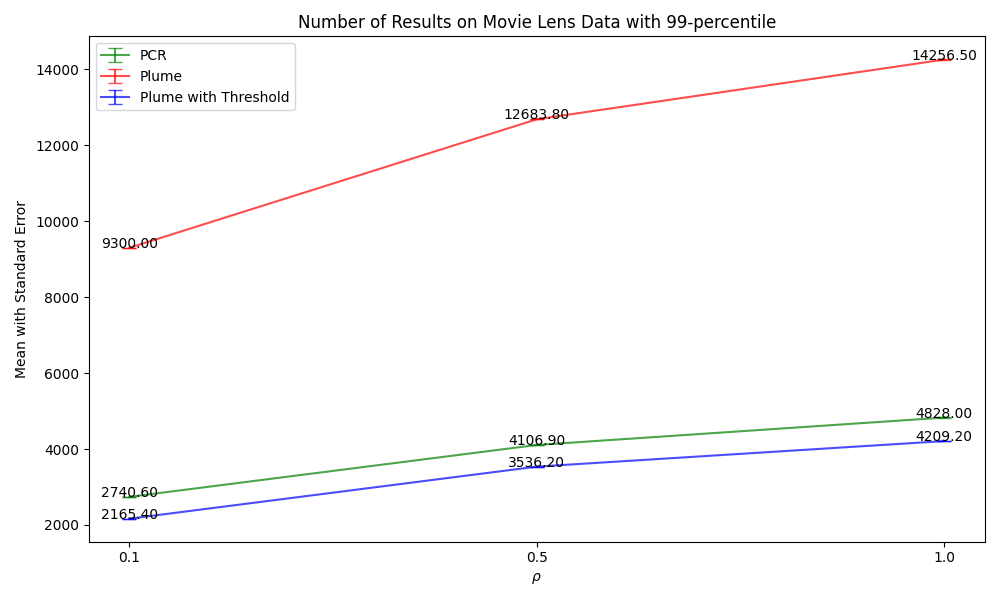}
\includegraphics[width=0.45\textwidth]{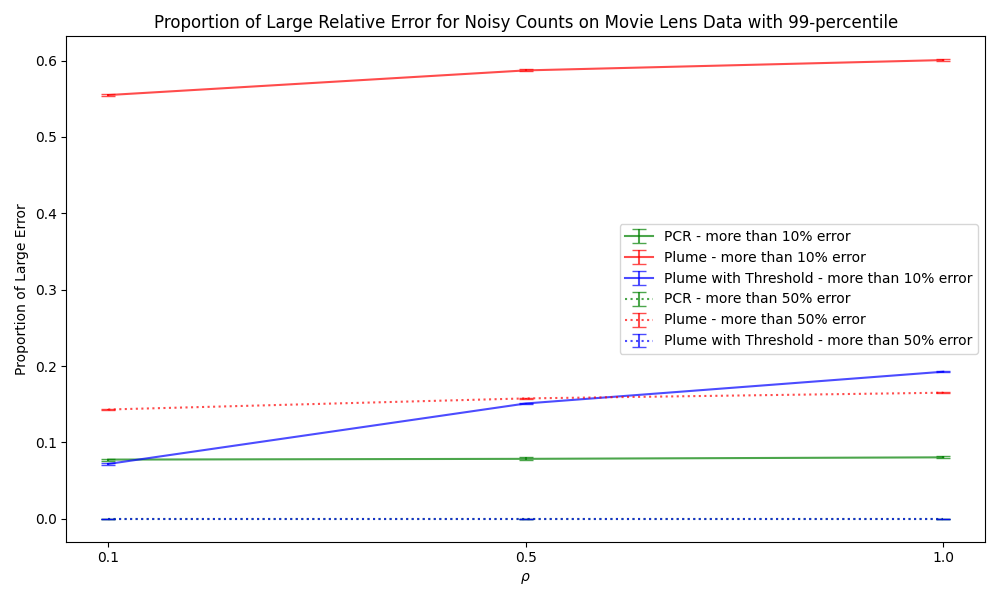}
\caption{\label{fig:Movie} Results for three approaches: PCR, Plume, and Plume with Threshold on the MovieLens data.  We show recall and precision for $\rho \in \{0.1, 0.5, 1.0 \}$ and $\delta = 10^{-6}$ averaged over 10 independent trials. The top plots use the true 95th-percentile for contribution bounding in Plume and the bottom plots use the true 99th-percentile for contribution bounding in Plume.}
\end{figure}

Note that the number of contributions per person is high in this dataset, so we see that PCR can return more results than Plume with a threshold and PCR has fewer counts with more than the target relative error when using either the true 95th or the 99th percentile.

\section{Conclusion}
We have presented an approach for releasing accurate counts on several public datasets that can outperform other approaches that rely on bounding the number of contributions each user makes to the dataset.  The accuracy of Plume crucially relies on setting the correct contribution bound, which itself should be computed with differential privacy, which has its own hyperparameters, such as allocating a certain percentage of the privacy budget for computing percentiles.  Although contribution bounding seems necessary for differential privacy, our approach (PCR) provides a way to release many accurate counts without bounding distinct contributions per user while ensuring a fixed differential privacy guarantee.   PCR can be used in a blackbox way, with little to no hyperparameter tuning, and can even provide better results than systems based on the approach in Plume especially when the number of contributions per user is high.  Such an approach is important in scenarios where direct access to the data is not allowed, even to those that are trying to implement differential privacy, and a decent baseline approach should be used.  We can simply use PCR on various real-world datasets that differ in size, as we showed in our results.  PCR would then allow practical deployments of differential privacy to scale across many different use cases, with minimal onboarding.  The setting where Plume seems to outperform PCR is when the number of distinct contributions per user is small, in which case directly using a Laplace or Gaussian mechanism over positive counts and releasing only noisy counts above a threshold \cite{KorolovaKeMiNt09, WilsonZhLaDeSiGi19} would perform better.  

There are other improvements one can make with PCR, including using noise reduction algorithms such as the Brownian mechanism \cite{WhitehouseWuRaRo22} to determine the level of noise for each count that is to be returned in cases where the true count is much larger than the threshold used in Unknown Domain Gumbel.  Furthermore, PCR can be extended to other aggregation functions beyond distinct counts by still using Unknown Domain Gumbel on distinct counts to find the elements with the most individuals and then doing contribution bounding or clipping on the aggregate function for the selected item.

\clearpage 
\bibliography{bib2}
\bibliographystyle{abbrvnat}

\end{document}